\documentclass{scrartcl}

\usepackage{microtype}

\usepackage[sort,numbers]{natbib}
\usepackage{fullpage}
\usepackage{hyperref}
\usepackage{amsthm}

\usepackage{amsmath}
\usepackage{amssymb}
\usepackage{comment}
\usepackage{amsfonts}

\usepackage{tikz}

\usepackage{placeins}

\usepackage{etoolbox}

\usepackage{tabularx,booktabs,multirow}
\usepackage{longtable}
\usepackage{pgfplots}
\usepackage{siunitx}
\usepackage{pgfplotstable}

\usepgfplotslibrary{groupplots}
\pgfplotsset{compat=1.14}

\usepackage[linesnumbered, ruled, procnumbered]{algorithm2e}
\DontPrintSemicolon

\makeatletter
\AtBeginEnvironment{function}{%
  \let\c@algocf\c@function
}
\makeatother

\usepackage[capitalize]{cleveref}
\crefname{cor}{Cor.}{Cor.}
\Crefname{cor}{Corollary}{Corrollaries}

{\bfseries}{\normalfont}
\theoremstyle{plain}
\newtheorem{theorem}{Theorem}{\bfseries}{\normalfont}
\newtheorem{obs}[theorem]{Observation}{\bfseries}{\normalfont}
{\bfseries}{\normalfont}
\newtheorem{prop}[theorem]{Proposition}{\bfseries}{\normalfont}
\newtheorem{lemma}[theorem]{Lemma}{\bfseries}{\normalfont}

\theoremstyle{definition}
\newtheorem{definition}{Definition}{\bfseries}{\normalfont}

\DeclareMathOperator{\poly}{poly}

\definecolor{ourblue}{RGB}{135,206,250}
\definecolor{ourgreen}{RGB}{0,100,0}
\definecolor{ourred}{RGB}{176,23,31}
\definecolor{lipicsyellow}{rgb}{0.99,0.78,0.07}

\usepackage{microtype,ellipsis}

\usepackage{units}
\usepackage[backgroundcolor=gray!10,disable]{todonotes}
\usepackage{scrtime}
\usepackage{eso-pic}

\usepackage[numbers,sort]{natbib}
\makeatletter
\def\NAT@spacechar{~}
\makeatother

\usepackage{paralist}

\usepackage{enumitem}

\graphicspath{{images/}}

 \newcommand{\appsymb}{$\star$}
 \newcommand{\appref}[1]{{\hyperref[#1]{\appsymb}}}

\usepackage{authblk}

\usepackage{macros}

\usepackage{tikz-cd}

\SetKwFunction{isVertexMaximal}{isVertexMaximal}
\SetKwFunction{isMaximal}{isMaximal}
\SetKwFunction{isolatedSubsets}{isolatedSubsets}

\title{Enumerating Isolated Cliques in Temporal Networks}

\author{Hendrik~Molter\thanks{Supported by the DFG, project MATE (NI 369/17).}}

\newcommand\CoAuthorMark{\footnotemark[\arabic{footnote}]} 

\author{Rolf~Niedermeier}

\author{Malte~Renken\protect\CoAuthorMark} 

\affil{\small Algorithmics and Computational Complexity, Fakult\"at
IV, TU Berlin, Germany, \\
 \texttt{\{h.molter, rolf.niedermeier, m.renken\}@tu-berlin.de}}

\date{}

\begin{document}

\maketitle

\begin{abstract}
Isolation is a concept from the world of clique enumeration that is mostly used to model communities that do not have much contact to the outside world.
Herein, a clique is considered \emph{isolated} if it has few edges connecting it to the rest of the graph.
Motivated by recent work on enumerating cliques in temporal networks,
we lift the isolation concept to this setting.
We discover that the addition of the time dimension
leads to six distinct natural isolation concepts.
Our main contribution is the development of fixed-parameter enumeration algorithms for five of these six clique types
employing the parameter ``degree of isolation''.
On the empirical side, we implement and test these algorithms
on (temporal) social network data, obtaining encouraging
preliminary results. 


\medskip


\end{abstract}

\section{Introduction}\label{sec:intro}
\begin{quote}
``Isolation is the one sure way to human happiness.''
\quad \quad -- Glenn Gould

\end{quote}
Clique detection and enumeration is a fundamental primitive of 
complex network analysis. In particular, there are numerous approaches 
(both from a more theoretical and from a more heuristic side) 
for listing all maximal cliques (that is, fully-connected subgraphs) in 
a graph.%
\footnote{\emph{Network} and \emph{Graph} are used interchangeably}
It is well-known that finding 
a maximum-cardinality clique is computationally hard (NP-hard, hard in the 
approximation sense and hard when parameterized by the 
clique cardinality).  Hence, heuristic approaches usually govern 
computational approaches to clique finding and enumeration. 
{F}rom now on, we focus on the case of enumerating maximal 
cliques. There have been numerous efforts to 
provide both theoretical guarantees and practically useful 
algorithms~\cite{eppstein2013listing,HuffnerKMN09,Ito2009,MaxIsolatedCliques,Tomita2006}.
In particular, to simplify (in a computational sense) the task on the one 
hand and to enumerate more meaningful maximal cliques (for specific application contexts)
on the other hand, \citet{Ito2009} 
introduced and investigated the enumeration of maximal cliques that 
are ``isolated''. Roughly speaking, isolation means that the 
connection of the maximal clique to the rest of the graph is 
limited, that is, there are few edges with one endpoint in the 
clique and one endpoint outside the clique; 
indeed, the degree of isolation  
can be controlled by choosing specific values of 
a corresponding isolation parameter. 
For instance, think of social networks where one wants to spot more or less
segregated sub-communities with little interaction to the world outside 
but intensive interaction inside the community.
We mention in passing that recently there have been (only) theoretical 
studies on the concept of ``secludedness''~\cite{Bev+18,fomin2017parameterized,chechik2017secluded} which is somewhat similar to the older isolation
concept: whereas for isolation one requests ``few outgoing edges'', 
for secludedness one asks for ``few outneighbors''; while finding 
isolated cliques becomes tractable~\cite{Ito2009}, finding secluded ones remains 
computationally hard~\cite{Bev+18}.

\citet{Ito2009} showed that in static networks isolated cliques can be 
enumerated efficiently; the only exponential factor in the running
time depends on the ``isolation parameter'', and so fairly isolated 
cliques can be  enumerated quite quickly. In follow-up work,
the isolation concept then was significantly extended and 
more thorough experimental studies (also with financial networks) 
have been performed~\cite{HuffnerKMN09,MaxIsolatedCliques}.
However, analyzing complex networks more and more means studying 
time-evolving networks. Hence, computational 
problems known from static networks also need 
to be solved on temporal networks (mathematically, these are graphs with 
fixed vertex set but a time-dependent edge set)~\cite{holme2012temporal,latapy2017stream,michail2016introduction}.
Thus, not surprisingly, the enumeration of maximal cliques has recently 
been lifted to the temporal setting~\cite{BHMMNS18,himmel17,viardCliqueTCS16,ViardML18}. While getting algorithmically more challenging than 
in the static network case, nevertheless the empirical results 
that have been achieved are encouraging. 
In this work, we now fill a gap by proposing 
to lift also the isolation concept to the 
temporal clique enumeration context, otherwise using the same modeling
of temporal cliques as in previous work.

Since we believe that enumerating isolated cliques has 
its most important applications in community detection scenarios, 
we focus on only two of three basic isolation concepts described by 
\citet{MaxIsolatedCliques} for the static setting.
More specifically, we only consider ``maximal isolation'' (every vertex has small outdegree) and ``average isolation'' (vertices have small outdegree on average),
but do not study ``minimal isolation'' (at least one vertex has small outdegree). 
Nevertheless we still face a richer modeling than in the static case since
isolation can happen in two ``dimensions'': vertices and time; 
for both we can consider maximum and average isolation.
With this distinction, we end up with eight natural 
ways to model isolation, where two ``pairs'' of isolation models turn out to be 
equivalent, finally leaving six different temporal 
isolation concepts for further study.

Our main contributions are as follows:
First, as indicated above, we already do some conceptual work 
with identifying six, mathematically formalized concepts of isolation
for temporal networks.  Second, building on and extending the algorithmic 
framework of \citet{MaxIsolatedCliques} for static networks, for small 
isolation values we provide efficient algorithms
for five of our six isolated clique enumeration models
and prove worst-case performance bounds for them.%
\footnote{In terms of the language of parameterized algorithmics, we show that these cases are fixed-parameter tractable when parameterized by isolation value.}
In this context, a main algorithmic contribution is the development of 
tailored subroutines (that are only partially shared between 
different isolation concepts). Finally, on the empirical side 
we contribute an encouraging first experimental analysis of 
our algorithms based on social network data. 
Our preliminary experiments
indicate significant differences (mostly in terms of practical running time) but
also (sometimes surprising) accordances between the concepts.

\section{Preliminaries}\label{sec:prelims}
In this section we first give some basic notation and terminology.
We then recall the isolation concept for static graphs and transfer it to temporal graphs.
Lastly, we give some motivating examples that are tailored to the different temporal isolation concepts arising
and try to give an intuitive understanding of the differences between the various temporal isolation models.

\paragraph{Static Graphs.}
Graphs in this paper are assumed to be undirected and simple.
To clearly distinguish them from temporal graphs, they are sometimes referred to as \emph{static} graphs.
Let $G=(V,E)$ be a static graph. We denote the vertex set of $G$ with $V(G)$ and the edge set of~$G$ with $E(G)$. For $v \in V(G)$ we use $\deg_G(v)$ for the number of edges ending at $v$.
For $v \in A \subseteq V$, $\outdeg_G(v, A)$ denotes the number of edges with one endpoint $v$ and the other one outside of $A$.
Further $\outdeg_G(A) := \sum_{v \in A} \outdeg_G(v, A)$.
We use $\delta_G(A) := \min_{v \in A} \deg_G(v)$ for the minimum degree.
In all these notations, we omit the index $G$ if there is no ambiguity.

\paragraph{Temporal Graphs and Temporal Cliques.}
A \emph{temporal graph} is a tuple $\TG = (V, E_1, E_2, \dots, E_\tau)$
of a vertex set $V$ and~$\tau$ edge sets $E_i \subseteq \binom{V}{2}$.
The graphs $G_i := \left(V, E_i \right)$ are called the \emph{layers} of~$\TG$.
The \emph{time edge set}~$\TE(\TG)$ (or $\TE$ if $\TG$ is clear from the context) is the disjoint union~$\coprod_{t=1}^\tau E_i$ of the edge sets of the layers of $\TG$.
For any $1 \leq a \leq b \leq \tau$ we define the (static) graphs $\bigcup_{t=a}^b G_t := \left(V, \bigcup_{t=a}^b E_t\right)$
and $\bigcap_{t=a}^b G_t := \left(V, \bigcap_{t=a}^b E_t\right)$.

Following the definition of \citet{viardCliqueTCS16}, a \emph{$\Delta$-clique} (for some $\Delta\in\NN$) of $\TG$ is a tuple $(C, [a, b])$  with $C \subseteq V$ and $1 \leq a \leq b \leq \tau$ such that $C$ is a clique in $\bigcup_{i=t}^{t+\Delta}G_i$ for all $t \in [a, b - \Delta]$.

It is easily observed that $(C, [a, b])$ is a $\Delta$-clique in $\TG$ if and only if $(C, [a, b-\Delta])$ is a 0-clique in~$\TG'=\left(V, E'_1, \ldots E'_{\tau-\Delta}\right)$ where $E'_i := \bigcup_{t=i}^{i+\Delta} E_t$.
Due to this, we will only concern ourselves with~$\Delta = 0$ and simply refer to 0-cliques as \emph{temporal cliques}.

\paragraph{Temporal Isolation.} We first introduce the isolation concepts for static graphs and then describe how we transfer them to the temporal setting. 
In a (static) graph $G$, a clique $C \subseteq V(G)$ is called \emph{avg-$c$-isolated} if $\outdeg_G(C) < c \cdot \abs{C}$
where $c \in \QQ$ is some positive number~\cite{Ito2009}.
Further it is called \emph{max-$c$-isolated} if $\max_{v \in C} \outdeg_G(v) < c$.
Clearly max-$c$-isolation implies avg-$c$-isolation.

Moving to temporal graphs, we want to define an isolation concept for temporal cliques. Recall that a temporal clique consist of a vertex set and a time interval.
We apply the isolation requirement both on a vertex and on a time level, meaning that for each dimension we can either require the average outdegree (as in the static avg-$c$-isolation) or the maximum outdegree (as in the static max-$c$-isolation) to be small.
To make this more clear, we give some examples. For instance, we can require that, on average over all layers, the maximum outdegree in a layer is small. 
Or we can require that the average outdegree must be small in every single layer. 
Note that the ordering of the requirements for the time dimension and the vertex dimension also matters. 
Requiring the average outdegree to be small in every layer is different from requiring that, on average over all vertices, the maximum degree over all time steps must be small. 
Having two isolation requirements (avg and max) for two dimensions with two possible orderings, we arrive at eight canonical temporal isolation types. 
However it turns out that if we use the same requirement for both dimensions, they behave commutatively, so it boils down to six \emph{different} temporal isolation types. 
In the following, we give a formal definition for each of the six temporal isolation types. 
To make the names less confusing, we use ``usually'' to refer to the avg isolation requirement in the time dimension and ``alltime'' to refer to the max isolation requirement in the time dimension.
\todo{understandable enough?}

\begin{definition}[Temporal Isolation]
Let $c\in\QQ$. A temporal clique $(C, [a, b])$ in a temporal graph $\TG=(V,E_1, \ldots, E_\tau)$ is called

\begin{itemize}
\item 
\emph{\alltimeavgisolated{c}} if 
$\max_{i \in [a,b]} \sum_{v \in C} \outdeg_{G_i}(v, C) < c \cdot \abs{C}$,

\item
\emph{\alltimemaxisolated{c}} if 
$\max_{v \in C} \max_{i \in [a,b]} \outdeg_{G_i}(v, C) < c$,

\item
\emph{\avgalltimeisolated{c}} if
$\sum_{v \in C} \max_{i \in [a,b]} \outdeg_{G_i}(v, C) < c \cdot \abs{C}$,

\item 
\emph{\maxoftenisolated{c}} if 
$\max_{v \in C} \sum_{i=a}^b  \outdeg_{G_i}(v, C) < c \cdot (b+1-a)$,

\item 
\emph{\oftenavgisolated{c}} if
$\sum_{i=a}^b \sum_{v \in C} \outdeg_{G_i}(v, C) < c \cdot \abs{C} \cdot (b+1-a)$, and

\item 
\emph{\oftenmaxisolated{c}} if
$\sum_{i=a}^b \max_{v \in C} \outdeg_{G_i}(v, C) < c \cdot (b+1-a)$.

\end{itemize}
\end{definition}

We define the set of \emph{isolation types} as $\Iii = \big\{$\alltimemax{}, \alltimeavg{}, \maxoften{}, \oftenavg{}, \avgalltime{}, \oftenmax{}$\big\}$.

For all isolation types $I \in \Iii$, an $I$-$c$-isolated temporal clique $(C, [a, b])$ is called \emph{time-maximal} if there is no other $I$-$c$-isolated clique
$(C, [a', b'])$ with $C' \supseteq C$ and $[a', b'] \supset [a, b]$.
If there is no $I$-$c$-isolated temporal clique $(C', [a, b])$ with $C' \supset C$, then we call $(C, [a, b])$ \emph{vertex-maximal}.
We call $(C, [a, b])$ \emph{maximal} if it is time-maximal and vertex-maximal.

Subsequently we try to give some intuition about the different isolation concepts.
Note that for sufficiently small $c$ they all converge to disallowing \emph{any} outgoing edges.
We start with the most restrictive and perhaps also most straightforward isolation type, that is, \textbf{\alltimemax{}-isolation}. 
Here all vertices are required to have little or no outside contact at all times---think of a quarantined group.
Slightly less restrictive is the notion of \textbf{\avgalltime{}-isolation}.
Here it would be possible to have some distinguished ``bridge'' vertices inside the clique with relatively much outside contact,
as long as most vertices never have many outgoing edges.
If we reorder the terms, we obtain \textbf{\alltimeavg{}-isolation}.
In contrast to the previous case, now the set of ``bridge'' vertices may be different at any point in time.
A typical situation where this could occur is that there is a low bandwidth connection between the clique and the rest of the graph,
only allowing a limited number of communications to occur at any given moment.
The next isolation concept, \textbf{\oftenmax{}-isolation}, can be seen as allowing short bursts of activity, in which some or even all vertices have many outgoing edges,
as long as the entire clique is isolated most of the time.
Again, if we reorder the terms, we get a less restrictive concept (\textbf{\maxoften{}-isolation}).
Here, the bursts of activity may happen at different times for different vertices.
Finally, \textbf{\oftenavg{}-isolation} is the least restrictive of these notions, 
only limiting the total number of outside contacts over all vertices and layers that are part of the temporal clique.

\section{Basic Facts}

We now prove some important facts that we will make use of in the correctness proofs of our algorithms.
Our first observation concerns the relation between different types of isolation.
It is easily checked using the definitions above.

\begin{obs}\label{thm:isolation_relations}
Let $\TG=(V,E_1, \ldots, E_\tau)$ be a temporal graph. The following implication diagram holds for any $a \leq b$, any clique $C$ in $\bigcap_{t=a}^b G_t$, and any $c > 0$:

\begin{center}
\tikzcdset{
  cells={font=\everymath\expandafter{\the\everymath\displaystyle}},
}

\begin{tikzcd}[row sep=2.5em, arrows=Rightarrow]
(C, [a, b]) \text{ \alltimemaxisolated{c}} \arrow[r] \arrow[d] & 
(C, [a, b]) \text{ \avgalltimeisolated{c}} \arrow[d] \\
(C, [a, b]) \text{ \oftenmaxisolated{c}} \arrow[d] &
(C, [a, b]) \text{ \alltimeavgisolated{c}} \arrow[d] \\
(C, [a, b]) \text{ \maxoftenisolated{c}}  \arrow[r] \arrow[d] &
(C, [a, b]) \text{ \oftenavgisolated{c}} \arrow[d]\\
C \text{ is max-$c$-isolated in } \bigcap_{i=a}^b G_i \arrow[r]
& C \text{ is avg-$c$-isolated in } \bigcap_{i=a}^b G_i
\end{tikzcd}

\end{center}
\end{obs}

Note that \cref{thm:isolation_relations} does not hold for \emph{maximal} isolated temporal cliques:
A maximal alltime-max-c-isolated clique is not necessarily a maximal usually-avg-c-isolated clique.

Next, we state two lemmata limiting the minimal size of isolated cliques,
helping us to confine the search space for our algorithms.
They are inspired by the ideas employed by \citet{MaxIsolatedCliques} in the static setting.

\begin{lemma}\label{thm:basic_lemma}
	Let $G$ be a static graph and let $C$ be a clique in $G$.
	Then, any avg-$c$-isolated subset $C' \subseteq C$ has size $\abs{C'} > \delta(C) - c + 1$.
\end{lemma}
\begin{proof}
	Suppose $\abs{C'} \leq \delta(C) - c +  1$.
	Then any vertex $w \in C'$ has
	$\outdeg(w, C') = \deg(w) - (\abs{C'} - 1) \geq \delta(C) - (\abs{C'}-1) \geq c$.
	Thus, $C'$ is not avg-$c$-isolated.
\end{proof}

\begin{lemma}\label{thm:lemma1b}
Let $C$ be a clique in $G_\cap := \bigcap_{i=1}^t G_i$ for some $G_i=(V,E_i)$.
Then any subset $C' \subseteq C$ for which $(C', [1, t])$ is \oftenavgisolated{c}
has size $\abs{C'} > \delta_{G_\cap}(C) - c + 1$.
\end{lemma}
\begin{proof}
Suppose not.
Then $\abs{C'} \leq \delta_{G_\cap}(C) - c + 1 \leq \delta_{G_i}(C) - c + 1$ for all $i$.
By \cref{thm:basic_lemma}, $C'$ is not avg-$c$-isolated in any $G_i$ and thus certainly not \oftenavgisolated{c}.
\end{proof}

Next, we show that some vertices must always be contained in vertex-maximal (and thus also in maximal) isolated cliques.
This will allow us to refrain from searching for maximal isolated cliques that do not contain these vertices.

\begin{lemma}\label{thm:lemma2b}
Let $C$ be a clique in $G_\cap := \bigcap_{i=1}^t G_i$ for some $G_i=(V,E_i)$ 
and let $C' \subseteq C$ be such that $(C', [1, t])$ is a vertex-maximal \oftenavgisolated{c} temporal clique.
Then $C'$ contains all vertices $v \in C$ that fulfill
\begin{align*}
\sum_{i=1}^t \deg_{G_i}(v) \leq t(\delta_{G_\cap}(C) + \abs{C'} + 1) \,.
\tag{$*$}
\end{align*}
\end{lemma}
\begin{proof}
We prove this statement by contraposition. Suppose $(*)$ holds for some $v \in C \setminus C'$.
Let~$k := \abs{C}$ and $k' := \abs{C'}$.
Then we have that
\begin{align*}	
\sum_{i=1}^t \outdeg_{G_i}(v, C) = 
\sum_{i=1}^t \left( \deg_{G_i}(v) - (k-1)\right) \leq t(\delta_{G_\cap}(C) - k + k' + 2) 
\end{align*}
and thus
\begin{align*}
\sum_{i=1}^t \outdeg_{G_i}(C'\cup \set{v}) &= 
\sum_{i=1}^t \left( \outdeg_{G_i}(C') + (k-k'-1) + \outdeg_{G_i}(v, C) - k' \right)
\\ &< t \left(c k' +  \delta_{G_\cap}(C) - k' + 1  \right)
\\&= t \left( c(k'+1) + \delta_{G_\cap}(C) -c -k' + 1 \right)
\\&< ct(k'+1),
\end{align*}
where the last inequality is due to \cref{thm:lemma1b}.
Thus, $C' \cup \{v\}$ is \oftenavgisolated{c}.
\end{proof}

\begin{lemma}\label{thm:highest_degree_vertices_b}
Let $C$ be a clique in $G_\cap := \bigcap_{i=1}^t G_i$ for some $G_i=(V,E_i)$ and let $C' \subseteq C$ such that $(C', [1, t])$ is a vertex-maximal \oftenavgisolated{c} temporal clique.
Let $\tilde{C} \subseteq C$ consist of the $\delta_{G_\cap}(C) - c + 2$ vertices $v$ with the lowest values of 
$
\sum_{i=1}^t \deg_{G_i}(v)
$.
Then $\tilde{C} \subseteq C'$.
\end{lemma}
\begin{proof}
Let $k := \abs{C}$ and $k' := \abs{C'}$.
For any $v \in C'$ and $i$, we have that
\begin{align*}
\outdeg_{G_i}(v, C') 
&\geq \outdeg_{G_\cap}(v, C')
\\ &= \outdeg_{G_\cap}(v, C) + k - k' 
\\ &\geq \delta_{G_\cap}(C) - k' + 1 =: d 
\,.
\end{align*}

Suppose for contradiction that there exists $u \in \tilde{C} \setminus C'$.
Then, for each $v \in C' \setminus \tilde{C}$, we have that
\begin{align*}	
\sum_{i=1}^t \outdeg_{G_i}(v, C')
&= \sum_{i=1}^t \left( \deg_{G_i}(v) - k' + 1 \right)
\\ &\geq \sum_{i=1}^t \left( \deg_{G_i}(u) - k' + 1 \right)
\\ &>  t \left( \delta_{G_\cap}(C)  + 2 \right) =: t\tilde{d},
\end{align*}
where the last inequality is due to \cref{thm:lemma2b}.
Furthermore, we have that 
\begin{align*}
h := \abs{C' \setminus \tilde{C}} 
&\geq \abs{C'} - (\abs{\tilde{C}} - 1) 
= k' - \delta_{G_\cap}(C) + c -1.
\end{align*}
Thus, we get that 
\begin{align*}
\sum_{i=1}^t \outdeg_{G_i}(C') &= \sum_{v \in C'} \sum_{i=1}^t \outdeg_{G_i}(v, C')
\\ &> k' td + h(t\tilde{d} - td)
\\ &= t \left( k'd + h( k' + 1) \right)
\\ &\geq t k' (d + h) 
\\ &= tk' c,
\end{align*}
contradicting the isolation of $C'$.
\end{proof}

\begin{lemma}
\label{thm:highest_degree_vertices}
Let $C$ be a clique and $C' \subseteq C$ a maximal avg-$c$-isolated subset.
Set $\tilde{C} \subseteq C$ to be the $\delta(C) - c + 2$ vertices of minimal degrees.
Then $\tilde{C} \subseteq C'$.
\end{lemma}
\begin{proof}
This is a special case of \cref{thm:highest_degree_vertices_b} for $t=1$.
\end{proof}

\section{Enumerating Maximal Isolated Temporal Cliques}

In this section, we present efficient algorithms to enumerate maximal isolated temporal cliques for five out of the six introduced temporal isolation concepts (all except \oftenmax{}).%
\footnote{The reader may wonder why \oftenmax{}-isolation was dropped here.
The answer is that, even though the same approach also works for \oftenmax{}-isolation,
we found no way to limit the work that would be required in the \isolatedSubsets{} subroutine significantly below $\bigO(2^n)$.}
These algorithms have \emph{fixed-parameter tractable} (FPT) running times for the isolation parameter~$c$,
that is, for fixed~$c$, the running time is a polynomial whose degree does not depend on $c$.%
\footnote{The isolation parameter~$c$ only influences the leading constant of the polynomial running time but not the degree of the polynomial,
that is, the running time is $f(c)\cdot\poly(|\TG|)$ for some function $f$.}
Formally, we show the following result, whose proof will be given in \cref{sec:correctness}.

\begin{theorem}\label{thm:FPT}
Let a temporal graph $\TG$, an isolation type $I\in \Iii \setminus \{$\oftenmax{}$\}$, and an isolation parameter $c\in\QQ$ be given,
then all maximal $I$-$c$-isolated temporal cliques in $\TG$ can be enumerated in FPT-time for the isolation parameter~$c$.
The specific running times depend on $I$ and are given in \cref{table:runningtimes}.
\end{theorem}

\begin{table}[t]
{
\caption{Running time of our maximal isolated temporal clique enumeration algorithms for the different temporal isolation types.}\label{table:runningtimes}
\renewcommand{\arraystretch}{1.5}
\begin{tabularx}{\textwidth}{X|X|X|X|X}
	
	\alltimeavg{} & \alltimemax{} & \avgalltime{} & \maxoften{} & \oftenavg{} \\
	\hline
	$\bigO(c^c \tau^2  \cdot \abs{V} \cdot \abs{\TE})$ & 
	$\bigO( 2.89^c  c \tau \cdot \abs{\TE})$ & 
	$\bigO( 5.78^c  c \tau  \cdot \abs{\TE})$  &
	$\bigO( 2.89^c  c \tau^3 \cdot \abs{\TE})$ &
	$\bigO( 5.78^c  c \tau^3 \cdot \abs{\TE})$\\
\end{tabularx}
}
\end{table}

Our algorithms are inspired by the algorithms for static isolated clique enumeration~\cite{Ito2009,MaxIsolatedCliques} and 
build upon the fact that every maximal $I$-$c$-isolated temporal clique $(C, [a, b])$ is contained in some vertex-maximal $c$-isolated clique $C'$ of $G_\cap := \bigcap_{t=a}^b G_t$
(by \cref{thm:isolation_relations}).
\Cref{alg:noIntersections} constitutes the top level algorithm.
Here, we iterate over all possible time windows $[a, b]$ and apply the so-called trimming procedure developed by \citet{Ito2009} to $G_\cap$ to obtain,
for each so-called \emph{pivot vertex} $v$, a set $C_v \subseteq N[v]$ containing all avg-$c$-isolated cliques of~$G_\cap$ that contain $v$. 
Subsequently, we enumerate all maximal cliques within $C_v$ and test each of them for maximal $I$-$c$-isolated subsets.
For this step, we employ \cref{thm:basic_lemma,thm:highest_degree_vertices_b,thm:highest_degree_vertices} to quickly skip over irrelevant subsets.
The details depend on the choice of $I$, as does the strategy for the last step, that is, 
removing non-maximal elements from the result set.
Remember that we have to pay attention to both, time- and vertex-maximality.
For the latter we can, in most cases, adapt an idea by \citet{MaxIsolatedCliques}.

We proceed by describing the subroutines \isolatedSubsets{} (\cref{line:isolatedSubsets} of \Cref{alg:noIntersections}) 
and \isMaximal{} (\cref{line:ismaximal} of \Cref{alg:noIntersections}). 
Then we prove correctness of our algorithms and, finally, analyze their running times.


\begin{algorithm}[t]\footnotesize
\caption{Enumerating maximal $I$-$c$-isolated cliques for $I \in \Iii \setminus \{$\oftenmax{}$\}$}
\label{alg:noIntersections}
\KwIn{
A temporal graph $\TG = (V, E_1, \ldots, E_\tau)$, a $c\in\QQ$, and an isolation type $I \in \Iii \setminus \{$\oftenmax{}$\}$.
}
\KwOut{
All maximal $I$-$c$-isolated cliques in $\TG$.
}
result $\gets \set{}$\;
\ForEach{$a = 1 \dots \tau$}{
	\ForEach{$b = a \dots \tau$}{
		\tcc{Here we are looking for cliques with lifetime $[a, b]$.}
		$G_\cap \gets \bigcap_{i=a}^b G_i$\;
		Sort vertices by ascending degree in $G_\cap$\;
		\ForEach{vertex $v$} { \tcc{Vertex $v$ is the pivot vertex.}
			$C_v \gets $ candidate set for pivot $v$ after trimming stage (in $G_\cap$)\;
			$k \gets \floor{\deg_{G_\cap}(v) - c + 2}$ \tcc{By \cref{thm:basic_lemma}, all isolated cliques are at least this large.}
			$\Cc \gets $ set of all maximal cliques of size at least $k$ in $C_v \subseteq G_\cap$\;
			\ForEach{$C \in \Cc$}{
				subsets $\gets I\text{-\isolatedSubsets{$C$, $[a, b]$, $\deg_{G_\cap}(v)$}}$\;\label{line:isolatedSubsets}
				result $\gets$ result ${} \cup  \set{(C, [a, b]) ; C \in \text{subsets}}$\;
			}
		}
	}
}\label{line:endfirstpart}
\ForEach{$(C, [a, b]) \in $ result}{\label{line:startsecondpart}
	\If{$I$-\isMaximal{$C$, $[a, b]$}}{\label{line:ismaximal}
		output $(C, [a, b])$\;
	}
}
\end{algorithm}

\subsection{Enumerating Isolated Subsets}

We now discuss the \isolatedSubsets{} subroutine of \Cref{alg:noIntersections} (\cref{line:isolatedSubsets}).
While the details depend on the isolation type, there are two main flavors.
For \alltimemax{}-isolation (Function \labelcref{alg:isolatedSubsets_alltimemax}) and \maxoften{}-isolation (Function \labelcref{alg:isolatedSubsets_maxoften}), it is possible to determine a single vertex that must be removed in order to obtain an isolated subset.
By repeatedly doing so, one either reaches an isolated subset or the size threshold set by \cref{thm:basic_lemma}.
In particular, each maximal clique contains at most one maximal isolated subset.

For \oftenavg{}-isolation (Function \labelcref{alg:isolatedSubsets_oftenavg}), \avgalltime{}-isolation (Function \labelcref{alg:isolatedSubsets_avgalltime}), and \alltimeavg{}-isolation (Function \labelcref{alg:isolatedSubsets_alltimeavg}), multiple vertices are removal candidates.
However, their number is upper-bounded by \cref{thm:highest_degree_vertices_b} and \Cref{thm:highest_degree_vertices}, respectively.
We therefore build a search tree, iteratively exploring removal sets of growing size.
The case of \alltimeavg{}-isolation (Function \labelcref{alg:isolatedSubsets_avgalltime}) is somewhat special, as here the set of removal candidates is different for each layer.

\begin{function}[t]\footnotesize
\caption{alltime-avg-isolatedSubsets($C$, ${[a, b]}$, $\delta$)}
\label{alg:isolatedSubsets_alltimeavg}
$d := \floor{\abs{C} - \delta + c - 2}$ \tcc{By \Cref{thm:highest_degree_vertices}, we can only remove the top $d$ vertices.}
$\Dd' \gets \set{\emptyset}$\;
result $\gets \emptyset$\;
\While{$\Dd' \neq \emptyset$}{
	$\Dd \gets \Dd'$\;
	$\Dd' \gets \emptyset$\;
	\ForEach{$D \in \Dd$}{
		$C' \gets C \setminus D$\;
		\If{$\exists i \in [a, b]: \sum_{v \in C'} \deg_{G_i}(v)  \geq \abs{C'}\cdot (\abs{C'}-1 + c)$}{
			take $i$ to be smallest possible\;
			\If{$\abs{C'} > \delta - c + 2$}{
				$E \gets $ the $d$ vertices of $C'$ that have the highest degrees in $G_i$\;\label{alltimeavg:12}
				$\Dd' \gets \Dd' \cup \set{D \cup \set{e} ; e \in E \setminus D}$\;\label{alltimeavg:13}
			}
		}
		\Else{
			result $\gets$ result${} \cup \set{C'}$\;	
		}
	}
}
\Return{result}\;
\end{function}

\begin{function}[t]\footnotesize
\caption{alltime-max-isolatedSubsets($C$, ${[a, b]}$, $\delta$)}
\label{alg:isolatedSubsets_alltimemax}
$\forall v \in C: s_v := \max_{i \in [a, b]}\deg_{G_i}(v)$\;
$k := \floor{\delta - c + 2}$ \tcc{By \cref{thm:basic_lemma}, all isolated cliques are at least this large.}
\While{$\abs{C} \geq k$}{
	\tcc{Remove offending vertices until we either succeed or fail.}
	\If{$\exists v \in C: s_v \geq \abs{C}-1 + c$} {
		$C \gets C \setminus \set{v}$\;
	}
	\Else{
		\Return{$\set{C}$}\;
	}
}
\Return{$\emptyset$}\;
\end{function}

\begin{function}[t]\footnotesize
\caption{avg-alltime-isolatedSubsets($C$, ${[a, b]}$, $\delta$)}
\label{alg:isolatedSubsets_avgalltime}
$\forall v \in C: s_v := \max_{i \in [a, b]} \deg_{G_i}(v)$\;
$d := \floor{\abs{C} - \delta + c - 2}$ \tcc{By \Cref{thm:highest_degree_vertices}, we can only remove the top $d$ vertices.}
$\set{v_i ; 1\leq i \leq d} := $ the $d$ vertices in $C$ with the highest values of $s_v$\;
$\Dd' \gets \{\emptyset\}$\;
result $\gets \emptyset$\;
\While{$\Dd' \neq \emptyset$}{
	$\Dd \gets \Dd'$\;
	$\Dd' \gets \emptyset$\;
	\ForEach{$D \in \Dd$}{
		$C' \gets C \setminus D$\;
		\If{$\sum_{v\in C'} s_v \geq \abs{C'} \cdot (\abs{C'}-1+c)$}{
			$j := \max\set{0, i ; v_i \in D}$\;
			$\Dd' \gets \Dd' \cup \set{D \cup \{v_i\} ; j < i \leq d}$\;
		}
		\Else{
			result $\gets$ result${} \cup \set{C \setminus D}$\;
		}
	}
}
\Return{result}\;
\end{function}

\begin{function}[t]\footnotesize
\caption{max-usually-isolatedSubsets($C$, ${[a, b]}$, $\delta$)}
\label{alg:isolatedSubsets_maxoften}
$\forall v \in C: s_v := \sum_{i \in [a, b]} \deg_{G_i}(v)$\;
$k := \floor{\delta - c + 2}$ \tcc{By \cref{thm:basic_lemma}, all isolated cliques are at least this large.}
\While{$\abs{C} \geq k$}{
	\tcc{Remove offending vertices until we either succeed or fail.}
	\If{$\exists v \in C: s_v  \geq (b-a+1)(\abs{C}-1 + c)$} {
		$C \gets C \setminus \set{v}$\;
	}
	\Else{
		\Return{$\set{C}$}\;
	}
}
\Return{$\emptyset$}\;
\end{function}

\begin{function}[t]\footnotesize
\caption{usually-avg-isolatedSubsets($C$, ${[a, b]}$, $\delta$)}
\label{alg:isolatedSubsets_oftenavg}
$\forall v \in C: s_v := \sum_{i \in [a, b]} \deg_{G_i}(v)$\;
$d := \floor{\abs{C} - \delta + c - 2}$ \tcc{By \cref{thm:highest_degree_vertices_b}, we can only remove the top $d$ vertices.}
$\set{v_i; 1 \leq i \leq d} := $ the $d$ vertices in $C$ with the highest values of $s_v$\;
$\Dd' \gets \{\emptyset\}$\;
result $\gets \emptyset$\;
\While{$\Dd' \neq \emptyset$}{
	$\Dd \gets \Dd'$\;
	$\Dd' \gets \emptyset$\;
	\ForEach{$D \in \Dd$}{
		$C' \gets C \setminus D$\;
		\If{$\sum_{v\in C'} s_v \geq (b-a+1) \cdot \abs{C'} \cdot (\abs{C'}-1 + c)$}{
			$j := \max\set{0, k ; v_k \in D}$\;
			$\Dd' \gets \Dd' \cup \set{D \cup \{v_k\} ; j < k \leq d}$\;
		}
		\Else{
			result $\gets$ result${} \cup \set{C \setminus D}$\;
		}
	}
}
\Return{result}\;
\end{function}

\subsection{Checking for Maximality}

We now discuss the \isMaximal{} subroutine of \Cref{alg:noIntersections} (\cref{line:ismaximal}), which in turn uses an \isVertexMaximal{} subroutine.
Note that, while each temporal clique $(C, [a, b])$ returned by \isolatedSubsets{} is vertex-maximal within its respective set $C_v$, it may be not vertex-maximal with regard to the entire graph.
Moreover, we need to check for maximality with regard to cliques with a larger time window.
The naive approach of pairwise comparing all elements of the result set is feasible but inefficient.
Instead, for \alltimemax{}-isolation, \alltimeavg{}-isolation, and \avgalltime{}-isolation it is sufficient to only check whether the time window can be extended in either direction (see Function \labelcref{alg:testMaximal_alltime}), and whether a larger clique exists within the same time window (i.e., checking vertex-maximality).
Except for the case of \alltimeavg{}-isolation (Function~\labelcref{alg:vertexMaxDumb}), 
the latter can again be implemented more efficiently than by using pairwise comparisons (Function~\labelcref{alg:vertexMaxClever}).
We modify the maximality test developed by \citet{MaxIsolatedCliques} which searches for cliques within the common neighborhood of $C$ and then checks whether these can be used to build a larger isolated clique.

This modified vertex-maximality test also works for the cases of \maxoften{}-isolation and \oftenavg{}-isolation, but here we cannot avoid checking all time windows
because isolation of $(C, [a, b])$ does not imply isolation of, say, $(C, [a, b-1])$ (Function \labelcref{alg:testMaximal_usually}).

\begin{function}[t]\footnotesize
\caption{$I$-isMaximal($C$, ${[a, b]}$) - version for $I \in \{\alltimeavg{},$ $\alltimemax{},$ $\avgalltime{}\}$}
\label{alg:testMaximal_alltime}
\For{$(a', b') \in \set{(a-1, b), (a, b+1)}$}{
	\If{$(C, [a', b'])$ is isolated clique}{
		\Return{$\false$}\;
	}
}
\Return{$I$-\isVertexMaximal{$C$, $[a, b]$}}\;
\end{function}

\begin{function}[t]\footnotesize
\caption{$I$-isMaximal($C$, ${[a, b]}$) - version for $I \in \{\maxoften{}, \oftenavg{}\}$}
\label{alg:testMaximal_usually}
\For{$a' = a...1$}{
	\If{$C$ is not a clique in $G_{a'}$}{
		break\;
	}
	\For{$b' = b...\tau$}{
		\If{$C$ is not a clique in $G_{b'}$}{
			break the inner loop\;
		}
		\If{$(C, [a', b'])$ $I$-$c$-isolated and $(a, b) \neq (a', b')$}{
			\Return{$\false$}\;
		}
		\If{not $I$-\isVertexMaximal{$C$, $[a', b']$}}{
			\Return{$\false$}\;
		}
	}
}
\Return{$\true$}\;
\end{function}

\begin{function}[t]\footnotesize
\caption{alltime-avg-isVertexMaximal($C$, ${[a, b]}$)}\label{alg:vertexMaxDumb}
\If{the result set contains any $(C', [a, b])$ with $C' \supset C$}{
	\Return{$\false$}
}
\Return{$\true$}
\end{function}

\begin{function}[t]\footnotesize
\caption{$I$-isVertexMaximal($C$, ${[a, b]}$) - version for $I \in \{\alltimemax{}, \avgalltime{}, \maxoften{}, \oftenavg{}\}$}\label{alg:vertexMaxClever}
$G_\cap := \bigcap_{i=a}^b G_i$\;
$w := \argmin_{v \in C}(\deg_{G_\cap}(v))$ \tcc{$w$ is the pivot of $(C, [a, b])$}
$S := \set{v \in N_{G_\cap}(w) \setminus C ; N_{G_\cap}(v) \supseteq C}$\;
$\Dd := $ set of all maximal cliques within $S \subset G_\cap$\;
\For{$D \in \Dd$}{
	\While{$(C \cup D, [a, b])$ not isolated}{
		\If{$I \in \{\maxoften{}, \oftenavg{}\}$}{
			$d \gets \argmax_{v \in D}(\sum_{i=a}^b \deg_{G_i}(v))$
		}
		\ElseIf{$I \in \{\alltimemax{}, \avgalltime{}\}$}{
			$d \gets \argmax_{v \in D}(\max_{i=a}^b \deg_{G_i}(v))$
		}
		$D \gets D \setminus \{d\}$\;
	}
	\If{$D \neq \emptyset$}{
		\Return false\;
	}
}
\Return true\;
\end{function}

\subsection{Correctness}\label{sec:correctness}

We now show the correctness of our algorithms. We first prove that the \isolatedSubsets{} functions (Functions \labelcref{alg:isolatedSubsets_alltimemax,alg:isolatedSubsets_maxoften,alg:isolatedSubsets_oftenavg,alg:isolatedSubsets_alltimeavg,alg:isolatedSubsets_avgalltime}) behave as intended. 

\begin{lemma}\label{thm:correctness_isolatedSubsets}
	Let $\TG = (V, E_1, \dots, E_\tau)$ be a temporal graph, $c\in\QQ$, and $I \in \Iii \setminus \{$\oftenmax{}$\}$. Let $C$ be a clique in $G_\cap := \bigcap_{i=a}^b G_i$ and $\delta = \delta_{G_\cap}(C)$.
	Then I-\isolatedSubsets{$C$, $[a, b]$, $\delta$} returns all maximal sets $\tilde{C} \subseteq C $ such that $(\tilde{C}, [a, b])$ is $I$-$c$-isolated.
\end{lemma}

\begin{proof}
	For the sake of brevity, we will simply write that some set $X \subseteq C$ is, say, \alltimeavg{}-isolated to denote that $(X, [a, b])$ is \alltimeavg{}-isolated.

	\textbf{Case 1: $I = $ \alltimeavg{} (Function \labelcref{alg:isolatedSubsets_alltimeavg}).}
	Let $\tilde{C} \subseteq C$ be any maximal subset which is \alltimeavgisolated{c},
	and suppose the algorithm is currently checking $C'$ with $\tilde{C} \subset C' \subseteq C$.
	Let~$i \in [a, b]$ be the first layer in which $C'$ is not avg-c-isolated.
	By \cref{thm:basic_lemma} we have that $\abs{C'} \geq \abs{\tilde{C}} + 1 > \delta(C) - c + 2$, thus the algorithm executes \cref{alltimeavg:12,alltimeavg:13}.
	Note that $\tilde{C}$ is avg-c-isolated in layer $i$, and let $\tilde{C}' \supseteq \tilde{C}$ be a maximal avg-c-isolated superset.
	Clearly $\tilde{C} \subseteq \tilde{C}' \subset C'$.
	By \Cref{thm:highest_degree_vertices}, we have that $C' \setminus \tilde{C}'$ is a subset of the set $E$ containing the $d$~highest-degree vertices of $C'$ in layer $i$.
	Consequently, the algorithm will add some set $C'' \subset C'$ with $C'' \supseteq \tilde{C}' \supseteq \tilde{C}$ to~$\Dd'$.
	By recursively applying the same argument to $C''$, we deduce that the algorithm will at some point reach $\tilde{C}$.

	\textbf{Case 2: $I = $ \alltimemax{} (Function \labelcref{alg:isolatedSubsets_alltimemax}).}
	By \cref{thm:basic_lemma} and \cref{thm:isolation_relations} we have that all \alltimemaxisolated{c} subsets of $C$ have at least size $k$.
	If $C$ contains an \alltimemaxisolated{c} subset $\tilde{C}$, then $\tilde{C}$ by definition does not contain any vertex $v$ with $s_v \geq \abs{\tilde{C}} - 1 + c$.
	Thus, by removing such vertices, we either reach the unique maximal \alltimemaxisolated{c} subset $\tilde{C}$ of~$C$, or, if we reach size $k$, may conclude that no such subset exists.
	
	\textbf{Case 3: $I = $ \avgalltime{} (Function \labelcref{alg:isolatedSubsets_avgalltime}).}
	Let $B := \set{v_i; 1 \leq i \leq d}$ and let $\tilde{C} \subseteq C$ be any maximal \avgalltimeisolated{c} subset.
	Note that any subset of $C$ is \avgalltimeisolated{c} if and only if the same set was avg-$c$-isolated in a static graph where each vertex' degree was set to $\max_{i\in [a, b]} \deg_{G_i}(v)$.
	By applying \Cref{thm:highest_degree_vertices} to this auxiliary graph, we see that $\tilde{C}$ must contain $C \setminus B$.
	Thus we observe analogously to Case 1 that the algorithm will at some point reach $\tilde{C}$.
	
	\textbf{Case 4: $I = $ \maxoften{} (Function \labelcref{alg:isolatedSubsets_maxoften}).}
	Works analogously to Case 2.
	
	\textbf{Case 5: $I = $ \oftenavg{} (Function \labelcref{alg:isolatedSubsets_oftenavg}).}
	Let $B := \set{v_i; 1 \leq i \leq d}$.
	By \cref{thm:highest_degree_vertices_b} we have that any maximal \oftenavgisolated{c} subset of $C$ must contain $C \setminus B$.
	Note that the loop generates all possible sets $D \subseteq B$ except those, for which $C \setminus B'$ has already found to be \oftenavgisolated{c} for some $B' \subset B$.
	Therefore, all maximal \oftenavgisolated{c} subsets of $C$ are added to the result set.
\end{proof}

Next, we prove that the function \isVertexMaximal{} (Functions \labelcref{alg:vertexMaxClever,alg:vertexMaxDumb}) behaves as intended.

\begin{lemma}\label{thm:correctness_isVertexMaximal}
	Let $\TG = (V, E_1, \dots, E_\tau)$ be a temporal graph, let $c\in\QQ$, and let $I \in \Iii \setminus \{$\oftenmax{}$\}$.
	Let $(C, [a, b])$ be an $I$-$c$-isolated clique in $\TG$.
	Then $I$-\isVertexMaximal{$C$, $[a, b]$} returns \true{} if and only if $(C, [a, b])$ is a vertex-maximal $I$-$c$-isolated clique.
\end{lemma}
\begin{proof}
	
	\textbf{Case 1: $I = $ \alltimeavg{} (Function~\labelcref{alg:vertexMaxDumb}).}
	In this case the algorithm simply performs a pairwise comparison of all cliques in this time window and is thus trivially correct.
	
	\textbf{Case 2: $I \in \{$\alltimemax, \avgalltime, \maxoften, \oftenavg$\}$ (Function~\labelcref{alg:vertexMaxClever}).}
	If the algorithm returns \false, then it has found a larger $I$-$c$-isolated clique $(C \cup D, [a, b])$.
	So suppose conversely that there is $C' \supset C$ for which $(C', [a, b])$ is an $I$-$c$-isolated clique.
	Then clearly $C' \subseteq C \cup S$ and thus also $C' \subseteq C \cup D$ for some $D \in \Dd$.
	Let $x := \abs{D \setminus C'} < \abs{D}$ and let $X \subset D$ be the set of the first $x$ vertices that the algorithm removes from $D$.
	Then it is not difficult to check for each of the four isolation types in question
	that $(C \cup D \setminus X, [a, b])$ is as least as $I$-isolated as $(C', [a, b])$.
	Thus the algorithm will not remove more than $x$ vertices from $D$ and instead return \false.
\end{proof}

Lastly, we show that the function \isMaximal{} (Functions \labelcref{alg:testMaximal_alltime,alg:testMaximal_usually}) behaves as intended.

\begin{lemma}\label{thm:correctness_isMaximal}
	Let $\TG = (V, E_1, \dots, E_\tau)$ be a temporal graph, let $c\in\QQ$, and let $I \in \Iii \setminus \{$\oftenmax{}$\}$. Let $(C, [a, b])$ be a $I$-$c$-isolated clique in $\TG$.
	Then $I$-\isMaximal{$C$, $[a, b]$} returns \true{} if and only if $(C, [a, b])$ is a maximal $I$-$c$-isolated clique.
\end{lemma}

\begin{proof}
	\textbf{Case 1: $I \in \{$\alltimemax{}, \alltimeavg{}, \avgalltime{}$\}$ (Function \labelcref{alg:testMaximal_alltime}).}
	By \cref{thm:correctness_isVertexMaximal} it only remains to show that the function returns \false{} if there exists an $I$-$c$-isolated clique $(C', [a', b'])$ with $a' < a$ or $b' > b$.
	Suppose without loss of generality that $a' < a$.
	Then $[a-1, b] \subseteq [a', b']$ and thus $(C', [a-1, b])$ is $I$-$c$-isolated.
	
	\textbf{Case 2: $I \in \{$\maxoften{}, \oftenavg{}$\}$ (Function \labelcref{alg:testMaximal_usually}).}
	Since the function systematically tries all possible time windows $[a', b'] \subseteq [a, b]$ for which $C$ is a clique, the correctness follows from \cref{thm:correctness_isVertexMaximal}.
\end{proof}

Now we have all the necessary pieces to prove the correctness of \Cref{alg:noIntersections}.

\begin{prop}[Correctness of \Cref{alg:noIntersections}]\label{prop:correctness}
	Let $\TG = (V, E_1, \dots, E_\tau)$ be a temporal graph, let $c\in\QQ$, and let $I \in \Iii \setminus \{$\oftenmax{}$\}$.
	Then, \Cref{alg:noIntersections} outputs exactly all maximal $I$-$c$-isolated temporal cliques.
\end{prop}

\begin{proof}
	By \cref{thm:correctness_isolatedSubsets} and \cref{thm:correctness_isMaximal} every element of the output is in fact a maximal $I$-$c$-isolated clique.
	So it remains to show that all such cliques are in fact found by the algorithm.
	To this end, let $(C, [a, b])$ be any maximal $I$-$c$-isolated clique.
	Then, $C$ is an avg-$c$-isolated clique in $G_\cap = \bigcap_{a \leq i \leq b} G_i$ by \cref{thm:isolation_relations}.
	\citet{Ito2009} showed that we then have $C \subseteq C_v$ where $v \in C$ is of minimum degree.
	Further $\abs{C} \geq \abs{C_v} - k$ by \cref{thm:basic_lemma}.
	Thus, $C \subseteq C'$ for some $C' \in \Cc$, so $(C, [a, b])$ is added to the result set by \cref{thm:correctness_isolatedSubsets}.
	Finally, $(C, [a, b])$ is also included in the output by \cref{thm:correctness_isMaximal}.
\end{proof}

\FloatBarrier{} 

\subsection{Running Time Analysis}

We will now estimate the time complexity of the different algorithms in terms of $\tau$, $\abs{V}$, $\abs{\TE}$, and the isolation parameter $c$.

We start estimating the running time of \Cref{alg:noIntersections} in terms of
$T^{(I)}_{\isolatedSubsets{}}$ and $T^{(I)}_{\isMaximal{}}$, which shall
denote the running times of the $I$-\isolatedSubsets{} and
$I$-\isMaximal{} subroutines, respectively, since they are the parts of the
running time that depend on the isolation type.

\begin{lemma}\label{lem:runningtime}
Let $\TG = (V, E_1, \dots, E_\tau)$ be a temporal graph, let $c\in\QQ$, and let
$I \in \Iii \setminus \{$\oftenmax{}$\}$. \Cref{alg:noIntersections} runs in
$\bigO\left(2^c c^2 \tau \cdot \abs{\TE} + T^{(I)}_{\isolatedSubsets{}} + T^{(I)}_{\isMaximal{}} \right)$ time.
\end{lemma}
\begin{proof}
We first investigate the running time of the first part of \Cref{alg:noIntersections} (the part up until \cref{line:endfirstpart}).
When iterating over all time windows $[a, b]$, each iteration of the inner loop extends the time window by one layer only.
Because of this, we can compute the intersection graphs in $\bigO(\tau \cdot \abs{\TE})$ time overall by using incremental updates.

All sorting steps of the algorithm can be done by bucketsort using constant time per time window and pivot vertex.
Computing $C_v$ for all vertices of $G_\cap$ takes $\bigO(c^3 \cdot \abs{E_\cap})$ time \cite[Lemmata~3.9 and~3.13]{Ito2009}.
Computing $\Cc$ from $C_v$ takes $\bigO(\abs{E(C_v)} + c \cdot \abs{C_v} + 2^c
c^2)$ time \cite[Proof of Proposition~1]{MaxIsolatedCliques}.
The overall number of steps (not counting \isolatedSubsets{} and \isMaximal{}) is thus
$\bigO\left(\tau \cdot \abs{\TE} + \sum_a \sum_b \left( c^3 \cdot \abs{E_\cap} + \sum_{\text{pivot }v} \left(  \abs{E(C_v)} + c \cdot \abs{C_v} + 2^c c^2 \right)\right)\right)$.

Note that $\sum_a \abs{E_\cap} \leq \sum_a \abs{E_a} \leq \abs{\TE}$.
Further, we assume that the algorithm is implemented to disregard vertices that have degree zero in~$G_a$ and thus also in $\bigcap_{t=a}^b G_t$.
Because of this assumption, we can also record the following observation.
If we sum $\deg_{G_a}(v) + 1$ over all time windows $[a, b]$ and pivot
vertices $v$, then the result is at most $\sum_a \sum_b
\sum_{\text{pivot } v} (\deg_{G_a}(v) + 1) \in \bigO(\sum_b
\sum_a \abs{E_a}) \subseteq \bigO(\tau \cdot \abs{\TE})$ by the handshake lemma.
Of course, the same estimation is valid when summing over any of the following:
$1 \leq \abs{C} \leq \abs{C_v} \leq \deg_{G_\cap}(v) +1 \leq \deg_{G_a}(v) +1$.

Another key observation is that $\sum_v \abs{E(C_v)} \in \bigO(c^3 \cdot
\abs{E_\cap})$ \cite[Lemma~3.13]{Ito2009}.
Using this, if we sum $\abs{E(C_v)}$ over all time windows and pivot vertices, we get at most
$\sum_a \sum_b \sum_{\text{pivot }v} \abs{E(C_v)}
\in \bigO(\tau \sum_a c^3 \cdot \abs{E_\cap}) 
\subseteq \bigO( c^3 \tau \cdot \abs{\TE})$.
Again, this also applies to $\abs{E(C)} \leq \abs{E(C_v)}$.

Employing these observations, the above running time can be bounded by
$\bigO(\tau \cdot \abs{\TE} + c^3 \tau \cdot \abs{\TE} +  +  c \tau \cdot
\abs{\TE} + 2^c c^2 \tau \cdot \abs{\TE}) \subseteq \bigO(2^c c^2 \tau \cdot
\abs{\TE})$.
\end{proof}

Now we analyze the running time $T^{(I)}_{\isolatedSubsets{}}$ of the
$I$-\isolatedSubsets{} subroutine (Functions
\labelcref{alg:isolatedSubsets_alltimemax,alg:isolatedSubsets_maxoften,alg:isolatedSubsets_oftenavg,alg:isolatedSubsets_alltimeavg,alg:isolatedSubsets_avgalltime})
 depending on the isolation type $I$.
 
\begin{lemma}\label{runningtime:isolatedsubsets}
$T^{(I)}_{\isolatedSubsets{}} \in 
\begin{cases}
\bigO(2^c c \tau \cdot \abs{\TE}) & \text{if } I \in \{\alltimemax{}, \maxoften{}\} , \\
\bigO(2^c c^2 \tau \cdot \abs{\TE}) & \text{if } I \in \{\oftenavg{}, \avgalltime{}\} , \\
\bigO(c^c \tau^2 \cdot \abs{\TE}) & \text{if } I \in \{\alltimeavg{}\} . \\
\end{cases}$

\end{lemma}
\begin{proof}
Keep in mind the observations made in the proof of \cref{lem:runningtime}, which are also useful here. 
Additionally, note that within any iteration of \Cref{alg:noIntersections}, $\Cc$ contains at
most $2^{\abs{C_v} - s}$ elements of size $s$ for each $k \leq s \leq \abs{C_v}$
and at most $2^{\abs{C_v} - k} \leq 2^c$ elements overall.

\textbf{Case 1: $I \in \{\alltimemax{}, \maxoften{}\}$.}
Computing $s_v$ takes $\bigO(\tau \cdot \abs{\TE})$ overall (again, using incremental updating between time windows).

For each call, the loop runs at most $\abs{C} - k \leq \abs{C} - (\delta + 1)  + c \leq c$ times, each needing $\bigO(\abs{C})$ time.
Since there are $\abs{\Cc}$ calls per time window and pivot,
the overall time is in
$\bigO(\tau \cdot \abs{\TE} + 2^c c \tau \cdot \abs{\TE}) \subseteq \bigO(2^c c
\tau \cdot \abs{\TE})$

\textbf{Case 2: $I = \alltimeavg{}$.}
There are $d^d$ possible options for $D$, each tested in $\bigO(\tau)$ time.
Of each size $s$ there are $2^{C_v -s}$ elements of that size within $\Cc$.
Note that $d \leq \abs{C} + c - \abs{C_v} -1 \leq c - 1$.
Thus, the loop needs 
\begin{align*}
	\bigO\left(\sum_s 2^{C_v - s} d^d \tau\right)
	\subseteq \bigO\left(\sum_s 2^{C_v - s} c^{s + c - C_v -1}\right)
	\subseteq \bigO\left( \sum_s c^{c-1} \tau\right)
	\subseteq \bigO\left( c^c \tau\right)
\end{align*}
time per time window and pivot,
giving at most $\bigO(c^c \tau^2  \cdot \abs{\TE})$ time overall.

\textbf{Case 3: $I \in \{\oftenavg{}, \avgalltime{}\}$.}
Computing $s_v$ again takes $\bigO(\tau \cdot\abs{\TE})$ time overall.

Here, there are $2^d$ possible options for $D$, each tested in constant time.
Thus, the loop needs
\begin{align*}
	\bigO\left(\sum_s 2^{C_v - s} 2^d \right)
	\subseteq \bigO\left(\sum_s 2^{C_v - s} 2^{s + c - C_v -1}\right)
	\subseteq \bigO\left( \sum_s 2^{c-1}\right)
	\subseteq \bigO\left( 2^c c\right)
\end{align*}
time per time window and pivot (again $s = \abs{C}$).
In total, this gives $\bigO(2^c c^2 \tau \cdot \abs{\TE})$.
\end{proof}

Finally, we analyze the running time $T^{(I)}_{\isMaximal{}}$ of the
\isMaximal{} subroutine (Functions \labelcref{alg:testMaximal_alltime,alg:testMaximal_usually})
 depending on the isolation type.
\begin{lemma}\label{runningtime:ismaximal}
$T^{(I)}_{\isMaximal{}} \in
\begin{cases}
	\bigO(2.89^c c \tau \cdot \abs{\TE}) & \text{if } I = \alltimemax{} , \\
	\bigO(2.89^c c \tau^3 \cdot \abs{\TE}) & \text{if } I = \maxoften{} , \\
	\bigO(2^{2c} \tau \cdot \abs{V} \cdot \abs{E} & \text{if } I = \alltimeavg{} , \\
	\bigO(5.78^c c \tau \cdot \abs{E} & \text{if } I = \avgalltime{} , \\
	\bigO(5.78^c c \tau^3 \cdot \abs{E} & \text{if } I = \oftenavg{} .
\end{cases}
$
\end{lemma}
\begin{proof}
\textbf{Case 1: $I = \alltimemax{}$.}
Each call to \isolatedSubsets{} returns at most one clique, thus
for each time window and pivot $v$ there are at most $\abs{\Cc} \leq 2^c$ cliques to be checked.
Each call to isMaximal takes $\bigO(\abs{C})$ time, in addition to one call to isVertexMaximal.

Regarding isVertexMaximal, the size of $S \subseteq N(v) \setminus C$ is at most $\deg(v) + 1 - \abs{C} < c$
and finding it takes $c \cdot \abs{C}$ time.
Computing $\Dd$ takes $\bigO(3^{c/3})$ time \cite{Tomita2006}
and $\Dd$ has size at most $3^{c/3}$.
For each $D \in \Dd$ we need $\bigO(\abs{D}) \subseteq \bigO(c)$ time.

Altogether, each call to isMaximal takes $\bigO(c \cdot  \abs{C} + 3^{c/3} c)$ time, giving an overall running time of
$\bigO(2^c c \tau  \cdot \abs{\TE} + 2^c 3^{c/3} c \tau  \cdot \abs{\TE}) \subseteq \bigO( 2.89^c c \tau \cdot \abs{\TE})$.

\textbf{Case 2: $I = \maxoften{}$.}
Again at most $2^c$ cliques need to be checked for each time window and pivot $v$.

For each call we need $\bigO(\tau \cdot \abs{E(C)})$ to determine the layers where $C$ is a clique and $\bigO(\tau^2 \cdot \abs{C})$ for the isolation check.
Further, there are $\tau^2$ calls to isVertexMaximal.
Of these, each takes $\bigO(3^{c/3}c)$ time as for \alltimemax{}-isolation.

Thus the total time per call is 
$\bigO(\tau\abs{E(C)} + \tau^2 \abs{C} + 3^{c/3}c\tau^2 )$,
giving an overall time bound of
$\bigO( 2^c c^3 \tau^2 \cdot \abs{\TE} +  2^c  \tau^3 \cdot \abs{\TE} +  2^c 3^{c/3} c \tau^3 \cdot \abs{\TE} )
\subseteq \bigO( 2.89^c c  \tau^3 \cdot \abs{\TE})$.

\textbf{Case 3: $I = \alltimeavg{}$.}
Each call to \isolatedSubsets{} returns at most $2^d \leq 2^c$ cliques,
therefore there are at most $2^{2c}$ cliques to be checked per time window and pivot.
Each call to isMaximal takes $\bigO(\abs{C})$ time, in addition to one call to isVertexMaximal.

Within isVertexMaximal, we only need to check against cliques for the same time window, 
of these there are at most $\sum_v \abs{\Cc}$ many, each of size at most $\abs{C_v}$.
The time needed to check a clique for maximality is linear in the total size of this set, i.e.\ 
$\bigO(\sum_v \abs{\Cc} \cdot \abs{C_v}) \subseteq \bigO(\abs{\Cc} \cdot \abs{E_\cap})$.

In total, each call to isMaximal takes $\bigO(\abs{\Cc} \cdot \abs{E_\cap})$ time, 
and the total time taken is thus $\bigO( 2^{2c} \tau \cdot \abs{V} \cdot \abs{\TE})$.

\textbf{Case 4: $I = \oftenavg{}$.}
Each call to \isolatedSubsets{} returns at most $2^d \leq 2^c$ cliques.
Apart from this extra factor, the analysis is identical to the \maxoften{}-isolation case.
Thus the total time is $\bigO( 2.89^c 2^c c \tau^3 \cdot \abs{\TE}) \subseteq \bigO(5.78^c c \tau^3  \cdot \abs{\TE})$.

\textbf{Case 5: $I = \avgalltime{}$.}
Each call to \isolatedSubsets{} returns at most $2^d \leq 2^c$ cliques.
Apart from this extra factor, the analysis is identical to the \alltimemax{}-isolation case.
Thus the total time is $\bigO( 2.89^c 2^c c \tau \cdot \abs{\TE}) \subseteq \bigO(5.78^c c \tau \cdot \abs{\TE})$.
\end{proof}

Now it is straightforward to check that
\cref{runningtime:isolatedsubsets,runningtime:ismaximal} together with
\cref{lem:runningtime} imply the running times given in 
\cref{table:runningtimes}. This together with \cref{prop:correctness} completes
the proof of \cref{thm:FPT}.

\section{Experimental Evaluation}\label{sec:exp}
In this section, we empirically evaluate the running times of our enumeration algorithms for
maximal isolated temporal cliques (\Cref{alg:noIntersections}) on several
real-world temporal graphs. In particular, we investigate the effect of different isolation
concepts as well as different values for isolation parameter~$c$ and~$\Delta$
(see the definition of $\Delta$-cliques in \cref{sec:prelims}) on the running time and on the number
of cliques that are enumerated.
We also draw some comparisons concerning running times to a state-of-the-art algorithm to enumerate
maximal (non-isolated) temporal cliques by \citet{BHMMNS18}.

\subsection{Setup and Statistics}
We implemented our algorithms%
\footnote{The code of our implementation is freely available at https://www.akt.tu-berlin.de/menue/software/}
in Python~3.6.8 and carried out experiments on an Intel Xeon E5-1620 computer clocked at 3.6\,GHz and with 64\,GB RAM running Debian GNU/Linux~6.0.
The given times refer to single-threaded computation.
\citet{BHMMNS18} implemented their algorithm in Python~2.7.12.

For the sake of comparability we tested our implementation on four freely available
data sets, three of which were also used by \citet{BHMMNS18}:
\begin{itemize}
\item Face-to-face contacts between high school students (``highschool-2011'', ``highschool-2012'', ``highschool-2013''~\cite{gemmetto2014mitigation,stehle2011high,fournet2014contact}), 
\item Spatial proximity between persons in a hospital (``tij\_pres\_LH10''~\cite{Genois2018}).
\end{itemize}

\begin{table}[t]
\footnotesize
  \centering
  \caption{Statistics for the data sets used in our experiments. The lifetime $\tau$ of a graph is the difference between the largest and smallest time stamp on an edge in the graph. The resolution $r$ indicates how often edges were measured. 
  }
  \begin{filecontents*}{GraphData.csv}
Data,Vertices,Edges,Lifetime (s),Resolution (s)
highschool-2011,126,28560,272330,20
highschool-2012,180,45047,729500,20
highschool-2013,327,188508,363560,20
tij\_pres\_LH10,73,150126,259180,20

\end{filecontents*}
\pgfplotstabletypeset[
   col sep=comma,
   columns={Data,Vertices,Edges,Resolution (s), Lifetime (s)},
   columns/Data/.style={column type=l,string type, column name={Data Set}},
   columns/Vertices/.style={column type=r, column name={\# Vertices $|V|$}},
   columns/Edges/.style={column type=r,int detect, column name={\# Edges $|\TE|$}},
   columns/Lifetime (s)/.style={column type=r,int detect, column name={Lifetime $\tau$ (in s)}},
   columns/Resolution (s)/.style={column type=r,int detect, column name={Resolution $r$ (in s)}},
   every head row/.style={before row=\toprule , after row=\midrule},
   every last row/.style={after row=\bottomrule},
   ]{GraphData.csv}
  \label{tab:stats1}
\end{table}

We list the most important statistics of the data set in \cref{tab:stats1}. We chose five roughly exponentially increasing values
$\eps,1,5,25,125$ for the isolation parameter $c$, where $\eps := 0.001$ effectively requires complete isolation and $125 \approx \abs{V}$ imposes little or no restriction.
We chose our~$\Delta$-values in the same fashion as \citet{BHMMNS18}. In order to limit the influence of time scales in the data and to make running times comparable between instances, 
the chosen $\Delta$-values of $0$, $5^3$, and $5^5$ were scaled by~$L/(5 \cdot \abs{\TE})$,
where $L$ is the temporal graph's lifetime in seconds~\cite[Section~5.1]{himmel17}.

\subsection{Experimental Results}

In \cref{fig:plot_highschool2011,fig:plot_highschool2012,fig:plot_highschool2013,fig:plot_hospital} 
the number of maximal isolated temporal cliques and the running time are plotted for each of the five isolation types and a range of isolation values $c$.
Missing values indicate that the respective instance exceeded the time limit of 1 hour.
In general, the different isolation types produce surprisingly similar output.
This suggests that the degrees of the vertices forming an isolated temporal clique are typically rather similar and remain constant over the lifetime of the clique.
Unsurprisingly, raising the value of $c$ increases the number of cliques as the isolation restriction is weakened.
However, this effect ceases roughly at $c =5$. 
Increasing $c$ further does not produce additional cliques, suggesting that the vertices in temporal cliques we found in the data sets mostly have out-degree at most five.
Furthermore, we can generally observe that the number of cliques decreases with increasing values of $\Delta$, which might seem unexpected at first glance, but is a consequence from finding many small cliques (with both few vertices and short time intervals) for small $\Delta$-values that ``merge together'' for larger $\Delta$-values. This behavior is consistent across all data sets we investigated.

Regarding running time, our algorithm is generally slower than the non-isolated clique enumeration algorithm by \citet{BHMMNS18}, even for small values of $c$. For comparison, the algorithm by \citet{BHMMNS18} solved the instances ``highschool-2011'', ``highschool-2012'', and ``highschool-2013'' for the same values for $\Delta$ that we considered in less than 17 Seconds per instance.
We believe that the two main reasons for our algorithm to be slower are the following. On the one hand, the maximality check we perform is much more complicated that the one of the algorithm of \citet{BHMMNS18}, which is an issue that also occurs in the static case~\cite{MaxIsolatedCliques,HuffnerKMN09}. On the other hand, we have to explicitly interate through more or less all possible intervals in which we could find an isolated temporal clique, which seems unavoidable in our setting. A particular consequence of this is that our algorithm is not \emph{output sensitive}, that is, the running time can be much larger than the number of maximal isolated temporal cliques in the input graph. In the case of (non-isolated) temporal clique enumeration, there are ways to circumvent these issues and in particular, the algorithm of \citet{BHMMNS18} is output sensitive. 
Both algorithms have a similar running time behavior with respect to $\Delta$, that is, the running time increases with $\Delta$, once $\Delta$ reaches moderately large values. 
Since higher values of $\Delta$ create a more dense graph after the preprocessing step, this behavior is expected. The algorithm of \citet{BHMMNS18} is slow for very small values of $\Delta$ that are close to zero (compared to itself for larger values of $\Delta$). We do not observe this phenomenon in most of our algorithms. In the variants for max-usually and usually-avg, however, we experience a similar issue for small values of $c$, especially visible in the ``tij\_pres\_LH10'' data set (\cref{fig:plot_hospital}), where the running time is surprisingly high for $\Delta=0$ and~$c=\varepsilon$. A possible explanation is that the usually-variants use a different maximality check than the alltime-variants, which may be the reason for this behavior.
Interestingly, no universal trend arises for the running time taken per resulting clique with respect to $c$, which stands in contrast to our theoretical running time analysis.

To get a more fine-grained picture, we tested intermediate values for $\Delta$ and $c$ on the ``tij\_pres\_LH10'' data for \avgalltime{}-isolation.
The results are shown in \cref{fig:delta_analysis} and \cref{fig:c_analysis}.
For increasing values of $\Delta$, the number of cliques drops while the running time per clique rises.
For fixed $\Delta$ and increasing $c$, the situation is very different.
Here, both number of cliques and running time per clique quickly rise and subsequently level off around $c=5$.

\definecolor{color0}{rgb}{0.60,0.60,0.60}
\definecolor{color1}{rgb}{0.82,0.39,0.18}
\definecolor{color2}{rgb}{0.15,0.44,0.38}
\definecolor{color3}{rgb}{0.69,0.68,0.95}

\pgfplotsset{every tick label/.append style={font=\tiny}}

\pgfplotsset{%
	myscatter/.style = {%
		only marks,
		mark size=3,
		mark options={solid}
	},
	marker1/.style = {%
		color1,
		mark=*
	},
	marker2/.style = {%
		color2,
		mark=triangle*
	},
	marker3/.style = {%
		color3,
		mark=square*
	},
	plotbyc/.style = {
		scaled y ticks=false,
		yticklabel style={/pgf/number format/fixed},
		xlabel style={align=center},
		width=0.27\textwidth,
		height=5cm,
		tick align=outside,
		tick pos=left,
		xmin=-0.3,
		xmax=4.3,
		xtick style={color0},
		xtick={0,1,2,3,4},
		xticklabels={$\eps$,1,5,25,125},
		axis line style={color0},
		ymode = log
	},
	numcliques/.style = {
		ymin = 100,
		ymax = 100000,
	},
	timeperclique/.style = {
		ymin = 0.002,
		ymax = 10,
	},
	contplot/.style = {
		width=0.5\textwidth,
		scaled y ticks=false,
		yticklabel style={/pgf/number format/fixed},
		height=5cm,
		tick align=outside,
		xtick style={color0},
		axis line style={color0},
		scale only axis,
		xtick pos=left,
	},
	timedata/.style = {
		color=green,
		solid,
		line width=2.0pt,
		mark=none,
	},
	cliquedata/.style = {
		color=blue,
		solid,
		line width=2.0pt,
		mark=none,
	},
}

\begin{figure}[p]
\begin{tikzpicture}
\footnotesize
\begin{groupplot}[
	plotbyc,
	group style={
		group size=5 by 2,
		horizontal sep=5pt,
		vertical sep=10pt
	},
]
\input{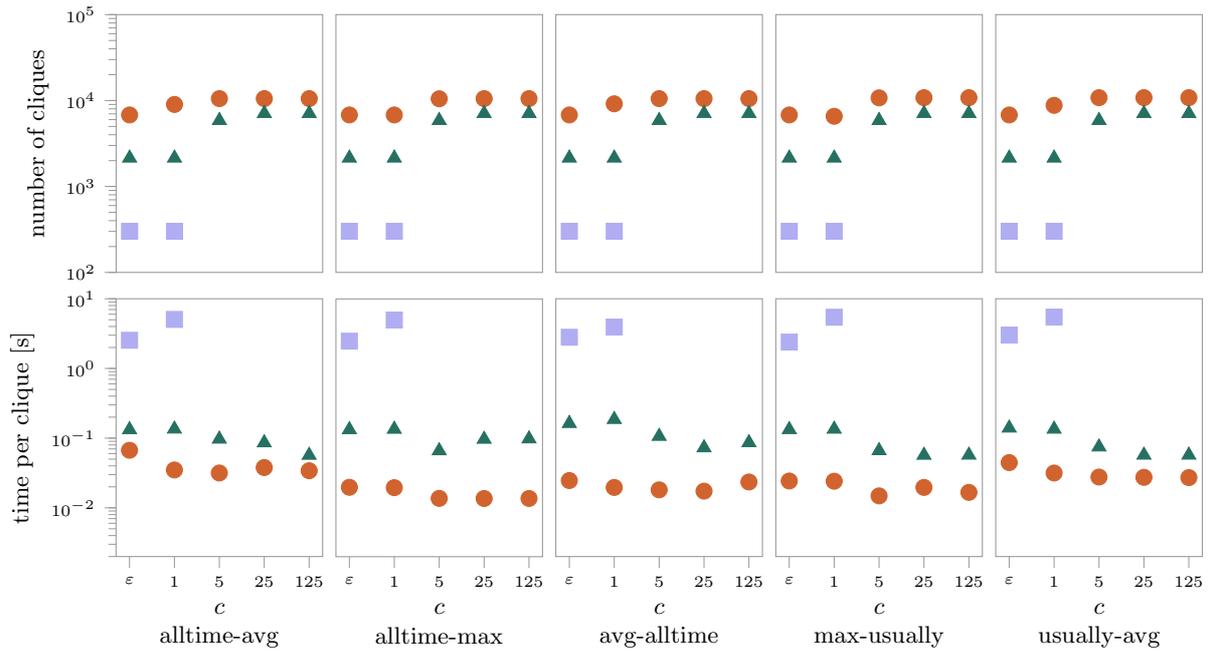}
\end{groupplot}
\end{tikzpicture}
\caption{Plot for the data set ``highschool-2011'' showing the number of cliques (top) and the computing time per clique (bottom) for the different temporal isolation types and different values of $c$ and $\Delta$.
The different $\Delta$-values are visualized by the different markers, with circles, triangles and squares denoting values of $0$, $5^3$, and $5^5$ respectively.}
\label{fig:plot_highschool2011}
\end{figure}

\begin{figure}[p]
\begin{tikzpicture}
\footnotesize
\begin{groupplot}[
	plotbyc,
	group style={
		group size=5 by 2,
		horizontal sep=5pt,
		vertical sep=10pt
	},
]
\input{plots/highschool2012.tex}
\end{groupplot}
\end{tikzpicture}
\caption{Plot for the data set ``highschool-2012'' (see also description of \cref{fig:plot_highschool2011}).}
\label{fig:plot_highschool2012}
\end{figure}

\begin{figure}[p]
\begin{tikzpicture}
\footnotesize
\begin{groupplot}[
	plotbyc,
	group style={
		group size=5 by 2,
		horizontal sep=5pt,
		vertical sep=10pt
	},
]
\input{plots/highschool2013.tex}
\end{groupplot}
\end{tikzpicture}
\caption{Plot for the data set ``highschool-2013'' (see also description of \cref{fig:plot_highschool2011}).}
\label{fig:plot_highschool2013}
\end{figure}

\begin{figure}[p]
\begin{tikzpicture}
\footnotesize
\begin{groupplot}[
	plotbyc,
	group style={
		group size=5 by 2,
		horizontal sep=5pt,
		vertical sep=10pt
	},
]
\input{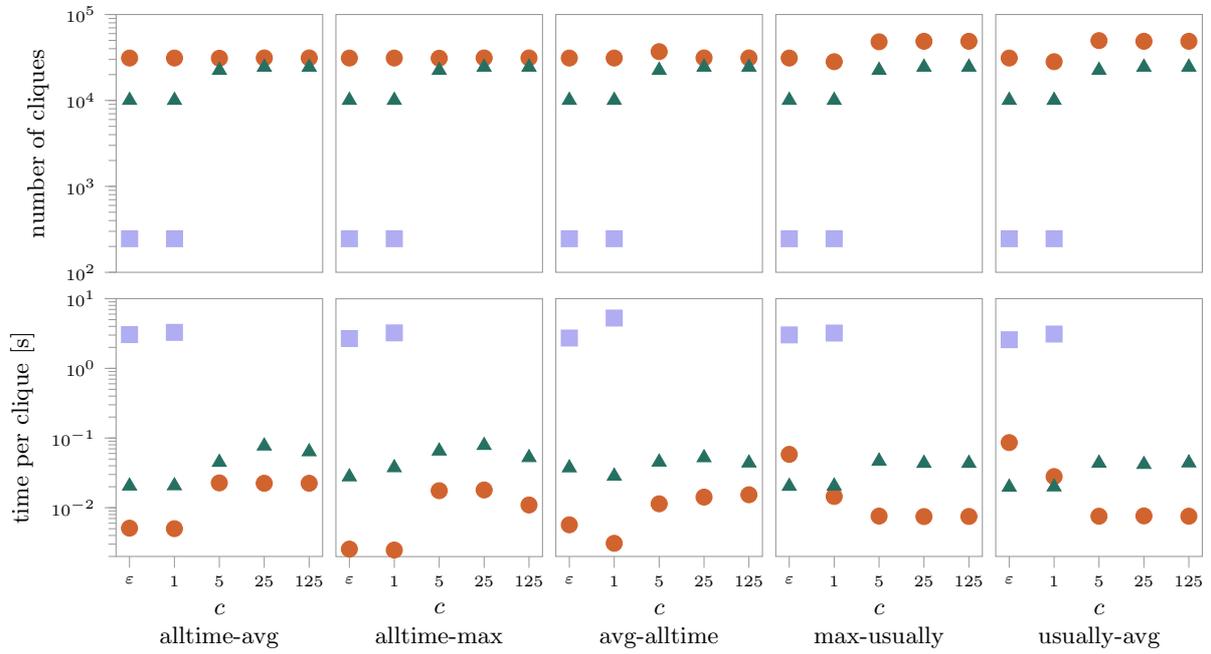}
\end{groupplot}
\end{tikzpicture}
\caption{Plot for the data set ``tij\_pres\_LH10'' (see also description of \cref{fig:plot_highschool2011}).}
\label{fig:plot_hospital}
\end{figure}

\begin{figure}[t]
\begin{minipage}{0.48\textwidth}
\begin{tikzpicture}
\begin{axis}[
	axis y line*=left,
	contplot,
	xmin=-60,
	xmax=780,
	ylabel={\ref*{data:time_by_delta} time per clique [s]},
	xlabel={$\Delta$}
]
\addplot[timedata]
table[
	col sep=comma,
	x={delta factor},
	y={time per clique},
] {consequtive_d_values_average_alltime.csv};
\label{data:time_by_delta}
\end{axis}
\begin{axis}[
	contplot,
	axis y line*=right,
	axis x line=none,
	xmin=-60,
	xmax=780,
	ylabel={\ref*{data:num_cliques_by_delta} number of cliques},
]
\addplot[cliquedata]
table[
	col sep=comma,
	x={delta factor},
	y={num_cliques},
] {consequtive_d_values_average_alltime.csv};
\label{data:num_cliques_by_delta}
\end{axis}
\end{tikzpicture}
\captionof{figure}{Number of cliques and running time for \avgalltime{}-isolation on the  data set ``tij\_pres\_LH10'' with $c=5$.}
\label{fig:delta_analysis}
\end{minipage}
\hfill
\begin{minipage}{0.48\textwidth}
\begin{tikzpicture}
\begin{axis}[
	axis y line*=left,
	contplot,
	xmin=0,
	xmax=36,
	ylabel={\ref*{data:time_by_c} time per clique [s]},
	xlabel={$c$},
]
\addplot[timedata]
table[
	col sep=comma,
	x={c},
	y={time per clique},
] {consequtive_c_values_average_alltime.csv};
\label{data:time_by_c}
\end{axis}
\begin{axis}[
	contplot,
	axis y line*=right,
	axis x line=none,
	xmin=0,
	xmax=36,
	ylabel={\ref*{data:num_cliques_by_c} number of cliques},
]
\addplot[cliquedata]
table[
	col sep=comma,
	x={c},
	y={num_cliques},
] {consequtive_c_values_average_alltime.csv};
\label{data:num_cliques_by_c}
\end{axis}
\end{tikzpicture}
\captionof{figure}{Number of cliques and running time for \avgalltime{}-isolation on the  data set ``tij\_pres\_LH10'' with $\Delta=5^3$.}
\label{fig:c_analysis}
\end{minipage}
\end{figure}

\section{Conclusion}\label{sec:concl}
We have lifted the concept of isolation from the static to the temporal setting, introducing six different types of temporal isolation.
For five out of those we developed algorithms and showed that enumerating maximal temporally isolated cliques is fixed-parameter tractable with respect to the isolation parameter.
This leaves one case (usually-max-isolation)  open for future research.

From an algorithm engineering perspective there is still room for improvement.
So far the practical running times make it hard to analyze larger data sets as done for example by \citet{BHMMNS18}.
Another possibility to approach this issue it to shift focus from the enumeration of all maximal temporally isolated cliques to the ``detection'' problem,
that is, to ``only'' search for one large temporally isolated cliques (if one exists). Depending on the application, this might still be a task worth investigating.
It could allow for better heuristic improvement such as pruning rules that remove parts of the input in which large cliques can be ruled out.

Finally, as in the static case, it would be natural to apply the isolation concepts to further community models such as for example temporal $k$-plexes~\cite{BHMMNS18}.

\paragraph*{Acknowledgments.} We want to thank our student assistant Fabian Jacobs for his work on the implementation of our algorithms.

\bibliography{bib}

\end{document}